\documentclass[a4paper,preprintnumbers,floatfix,superscriptaddress,pra,twocolumn,showpacs,notitlepage,longbibliography
]{revtex4-2}
\usepackage[T1]{fontenc}
\usepackage{inputenc}
\usepackage{color}
\usepackage{babel}
\usepackage{physics}
\usepackage{mathtools}
\usepackage{bbding}
\usepackage{pifont}
\usepackage{textcomp}
\usepackage{wasysym}
\usepackage{amsthm}
\usepackage[normalem]{ulem}
\usepackage{nicefrac}
\usepackage{wasysym}
\usepackage{dsfont}
\usepackage{amsmath}
\usepackage{amssymb}
\usepackage{graphicx}
\usepackage{hyperref}
\usepackage{amsfonts}
\usepackage{enumerate}
\hypersetup{
     colorlinks   = true,
     linkcolor    = blue,
     citecolor    = blue,
     urlcolor     = blue,
     }

\renewcommand{\selectlanguage}[1]{}
\makeatletter
\@ifundefined{textcolor}{}
{%
	\definecolor{BLACK}{gray}{0}
	\definecolor{WHITE}{gray}{1}
	\definecolor{RED}{rgb}{1,0,0}
	\definecolor{GREEN}{rgb}{0,1,0}
	\definecolor{BLUE}{rgb}{0,0,1}
	\definecolor{CYAN}{cmyk}{1,0,0,0}
	\definecolor{MAGENTA}{cmyk}{0,1,0,0}
	\definecolor{YELLOW}{cmyk}{0,0,1,0}
}
\theoremstyle{plain}

\theoremstyle{plain}

\ifx\proof\undefined
\newenvironment{proof}[1][\protect\proofname]{\par
	\normalfont\topsep6\p@\@plus6\p@\relax
	\Trivlist
	\itemindent\parindent
	\item[\hskip\labelsep
	\scshape
	#1]\ignorespaces
}{%
	\endtrivlist\@endpefalse
}
\providecommand{\proofname}{Proof}
\fi
\theoremstyle{plain}

\providecommand{\lemmaname}{Lemma}
\providecommand{\definitionname}{Definition}
\providecommand{\propositionname}{Proposition}

\usepackage{babel}
\usepackage{txfonts}
\usepackage{colortbl}\definecolor{myurlcolor}{rgb}{0,0,0.7}

\renewcommand{\ket}[1]{\left| #1 \right\rangle}

\renewcommand{\ketbra}[2]{\left|#1\middle\rangle\!\middle\langle#2\right|}
\newcommand{\proj}[1]{\ketbra{#1}{#1}}

\newcommand{\haH}

\newtheorem{theorem}{Theorem}

\newtheorem{lemma}{Lemma}

\newtheorem{definition}{Definition}

\usepackage {hyperref}
\hypersetup{
     colorlinks   = true,
     linkcolor    = blue,
     citecolor    = blue,
     urlcolor     = blue,
     }
\makeatother

\newcommand{\Ket}[1]{\left| #1 \middle \rangle\!\right \rangle}

\newcommand{\raisedtarget}[1]{%
  \raisebox{\fontcharht\font`P}[0pt][0pt]{\hypertarget{#1}{}}%
}

\begin{document}
\title{Entropic limitations on fixed causal order}

\date{\today}

\author{Matheus Capela}
\email{capela.quantum@gmail.com}
\affiliation{Departamento de F\'isica, Universidade Federal de S\~{a}o Carlos, Rodovia Washington Lu\'is, km 235 - SP-310, 13565-905 S\~{a}o Carlos, SP, Brazil}

\author{Kaumudibikash Goswami}
\email{goswami.kaumudibikash@gmail.com}
\affiliation{QICI Quantum Information and Computation Initiative, Department of Computer Science,
The University of Hong Kong, Pokfulam Road, Hong Kong}

\date{\today}
\begin{abstract}
Quantum processes can exhibit scenarios beyond a fixed order of events. We propose information inequalities that, when violated, constitute sufficient conditions to certify quantum processes without a fixed causal order---causally separable or indefinite causal ordered processes. The inequalities hold valid for a vast class of information measures. Nevertheless, we take under scrutiny the von Neumann, $\alpha-$R\'enyi entropies with parameter $\alpha \in [1/2,1) \cup (1,\infty)$, and max- and min-entropies. We also discuss how the strong subadditivity of quantum (von Neumann) entropy, used along with the information inequality developed here, implies relevant witnesses of causally separable and indefinite causal ordered processes in marginal scenarios. Importantly, we show the violation of these inequalities for the quantum switch, a paradigmatic example of a process with indefinite causal order. Our approach contributes to the important research direction of information-theoretic characterization of quantum processes beyond fixed causal orders. 
\end{abstract}
\maketitle
\section{Introduction}
\noindent

In our daily experience, we observe events in a fixed order. Such a notion of fixed causal order is well captured in the conventional framework of causal models~\cite{Pearl2009, Spirtes2001}. However, it is possible to have a more general scenario where the causal order has classical uncertainty. A typical example is a distributed network with each node associated with an event. In such a scenario, the random fluctuation of the local clocks at the nodes induces a probabilistic mixture of different causal orders~\cite{Lamport1978}. More interestingly, in quantum theory, it is possible to have indefinite causal order~\cite{oreshkov2012quantum, chiribella2013quantum} where the departure from the fixed causal order cannot be attributed to classical ignorance. Indefinite causal ordered processes are important from a foundational perspective of quantum gravity~\cite{Hardy_2007, Hardy_2009, Zych_2019} and from an applied aspect of computational and communication advantages~\cite{Ara_jo_2014,Gu_rin_2016,Ara_jo_2017,ebler2018enhanced,Chiribella_2021}.  To capture such general causal structures, the framework of supermaps has been proposed~\cite{chiribella2008transforming, chiribella2013quantum}, which can be represented via the process matrix~\cite{oreshkov2012quantum}. In this framework, the events, represented by quantum operations, are input to the supermaps, which encode all the background causal relationships. The fixed-ordered processes are a particular class of supermaps, known as quantum combs~\cite{chiribella2009theoretical}, or multi-time process~\cite{pollock2018non}. These processes physically signify the dynamics of open quantum systems where the quantum system of interest interacts with the environment, and the events are quantum operations on the system in between those system-environment interactions. 


An interesting question is how to demarcate such fixed causal ordered processes from causally separable and indefinite causal ordered processes. The crude approach of performing process tomography of the supermap quickly becomes intractable with increasing local dimensions and the number of events. In this work, we aim to give an alternative answer from an information-theoretic perspective.




The information-theoretic characterization of quantum processes has found important applications in quantum communication~\cite{wilde2011classical}, quantum cryptography~\cite{Renner2005}, quantum networks~\cite{gamal2011lecturenotes}, black hole physics~\cite{Bekenstein2004}, and quantum thermodynamics~\cite{Goold2016}. 
One of the fundamental principles of information theory is the data processing inequality, which states that the correlation between a system and a reference, quantified by mutual information, does not increase under the action of local quantum channels~\cite{Lieb1973, Schumacher1996}. Data processing inequalities are always respected for Markovian processes---a subset of the multi-time process where the environment does not retain any memory of the system. More generally, Markovian processes can exhibit information inequalities which are not equivalent to the standard data processing inequalities~\cite{capela2020monogamy,capela2022quantum}. For example, in Ref.~\cite{capela2022quantum}, violation of such general data processing inequalities sufficiently guarantees the dynamics are non-Markovian, even when the dynamics respect the standard data processing inequalities. Although entropic inequalities for Markovian processes have been explored in~\cite{capela2020monogamy,capela2022quantum}, similar entropic inequalities constituting a sufficient condition to certify processes beyond a fixed causal ordered process have been lacking. 

We bridge this gap by addressing the following questions: What information-theoretic conditions would emerge from the most general fixed causal ordered process? How well can those information inequalities certify processes which do not exhibit a fixed causal order? We answer this by developing a quantum information inequality which holds for every quantum non-Markovian process. The key ingredient to derive the entropic inequalities is as follows: first, we feed the subsystem of maximally entangled states to the process. With this construction, the reduced local channel on the environment guarantees that the input maximally mixed state maps to the output maximally mixed state. Now, as a consequence monotonicity of quantum relative entropy, the von Neumann entropy of the states evolving through this channel must be monotonically increasing with a suitable dimension-dependent lower bound. Using this fact, and considering that the von Neumann entropies are the same across arbitrary bipartitions of a pure state, we obtain the desired entropic inequalities involving only accessible systems and their local dimensions. Our proposed entropic inequality provides a sufficient condition for the certification of causally separable and indefinite causal ordered processes. As an example, we show the violation of the inequalities for an important class of indefinite causal ordered processes: the quantum switch~\cite{chiribella2013quantum}. 


Our work differs from the previously developed certification methods~\cite{oreshkov2012quantum,branciard2015simplest,dourdent2022semi,van2023device}, which are intended to identify indefinite causal ordered processes only. In contrast, our witnesses are sensitive to both indefinite causal order and causally separable processes. An example where such certification of causally separable processes could be relevant is when we want to identify a noisy quantum channel, due to an underlying causally separable process affecting the overall quantum capacity \cite{jia2019causal}. Moreover, our inequalities are entropic in nature, which can potentially lead to important operational tasks involving different causal structures, similar to how the entropic version of the Heisenberg uncertainty principle has been crucial in quantum cryptography~\cite{Coles2017}.

Currently, our entropic inequalities are device-dependent in nature. However, we open up possibilities of developing entropic semi-device-dependent or device-independent inequalities, identifying beyond fixed causal ordered processes, in a similar spirit to semi-device-dependent~\cite {dourdent2022semi} and device-independent~\cite {van2023device, Dourdent2024} causal inequalities.

The paper is organized as follows. In Sec.~\ref {sec_2}, we present the background on process formalism and information inequalities. In Sec.~\ref{sec_witnesses} we present the main result of the paper, the quantum data-processing conditions on fixed causal ordered processes. Moreover, we present applications of the information inequalities on the certification of indefinite causal order. In Sec .~\ref {sec_4}, we consider how the information inequality developed here, along with strong subadditivity of quantum entropy, implies useful witnesses of the departure from the fixed causal order from marginal information. In Sec .~\ref {sec_5}, we show how different entropy measures can be used in the development of similar data-processing inequalities. We focus our attention on the quantum R\'enyi entropies, thus generalizing our results presented in Sec.~\ref{sec_witnesses}. The use of R\'enyi entropies also implies advantages in the certification of departure from the fixed causal order in the quantum switch when compared with the von Neumann entropy. Finally, we present some remarks on future directions of research following the results in this paper in Sec.~\ref{sec_6}.

\section{Preliminaries} \label{sec_2}

\subsection{The quantum process formalism}
Every quantum system $A$ is associated with a Hilbert space $\mathcal{H}_A$. From now on, we denote the quantum system and its associated Hilbert space with the same letter.  We consider finite-dimensional Hilbert spaces only, and denote the dimension of a system $A$ with $\mathrm{dim}(A)$. The set of linear bounded operators on $A$ is denoted with $L(A)$. In case of multiple systems, $A_1, A_2, \cdots, A_n$, we represent the set as $L(A_1{\otimes}A_2{\otimes}\cdots{\otimes}A_n)$. The identity operator is denoted as $\mathds{1}$. The quantum state $\rho \in L(A)$ is a positive semidefinite operator with a unit trace. The set of linear transformations with input system $A$ and output system $B$ is denoted with $L(A,B)$. Any linear transformation $\mathcal{T} \in L(A,B)$ can be equivalently represented via its Choi-Jamio{\l}kowski (CJ) vector in $L(A {\otimes} B)$~\cite{jamio72, choi75,choi1975completely}:


\begin{equation} \label{Vec}
 \Ket{\mathcal{T}} = \sum_i \ket{i} \otimes \mathcal{T} \ket{i},
\end{equation}
where $\{\ket{i}\}$ is an orthonormal basis of the input system $A$.


The collection of linear maps $\Lambda:L(A) \rightarrow L(B)$ with input system $A$ and output system $B$ is denoted with $T(A,B)$. The identity map is denoted as $\mathrm{id}$. In turn, a linear map $\Lambda$ in $T(A,B)$ is isomorphic to linear operators in $L(A, B)$ due to its CJ representation~\cite{jamio72, choi75,choi1975completely}:

\begin{equation} \label{Choi}
    J(\Lambda)=\sum_{i,j} \ketbra{i}{j} \otimes \Lambda(\ketbra{i}{j}),
\end{equation}

A linear map $\Lambda:L(A)\to L(B)$ represents a quantum operation if and only if it is a completely positive and trace-nonincreasing (CPTNI) map. The complete positivity condition of $\Lambda$ is equivalent to the positivity of its CJ representation: $J(\Lambda) {\ge} 0$, and the trace-non-increasing condition is equivalent to $\Tr_{B}J(\Lambda)\le \mathds{1}_A$. A quantum channel, $\Lambda$, is a quantum operation which is a complete positive and trace preserving (CPTP) map, where the trace preserving condition is equivalent to $\Tr_{B}J(\Lambda)= \mathds{1}_A$.



Finally, a linear supermap is defined as a linear transformation $\Upsilon : T(A_0,A_1)\rightarrow T(P,F)$ mapping linear maps with input system $A_0$ and output system $A_1$ into linear maps with global past $P$ and global future $F$. The identity supermap is denoted as $\mathbf{id}$. A  bilinear supermap is similarly defined as a bilinear transformation $\Upsilon : T(A_0,A_1) \times T(B_0,B_1) \rightarrow T(P,F)$, i.e, it transforms a pair of linear maps into an output linear map. The definition of multilinear supermaps is a direct extension of the bilinear ones: a transformation of several input linear maps into an output linear map. 

\begin{figure}[t!] 
\centering   

\includegraphics[width=0.45\textwidth]{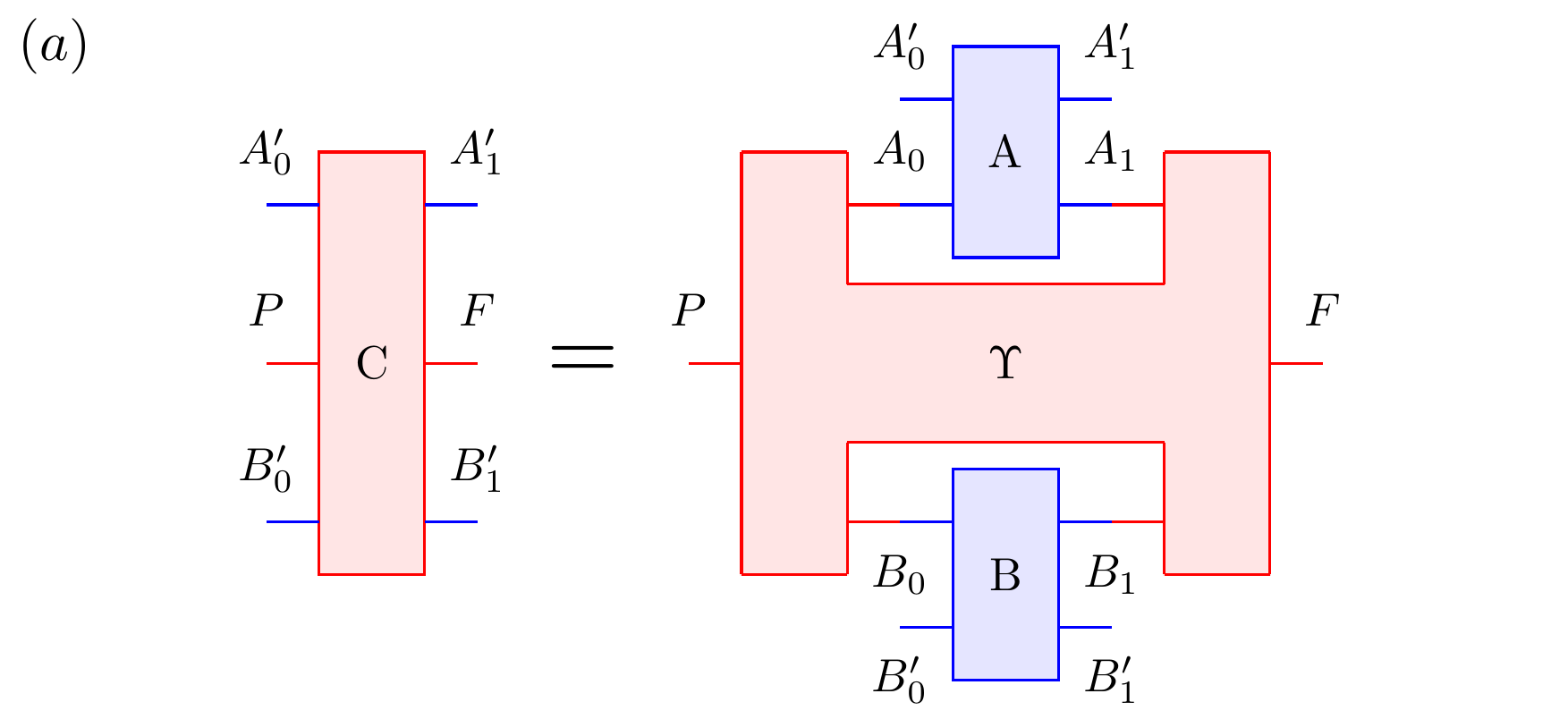} 

\includegraphics[width=0.45\textwidth]{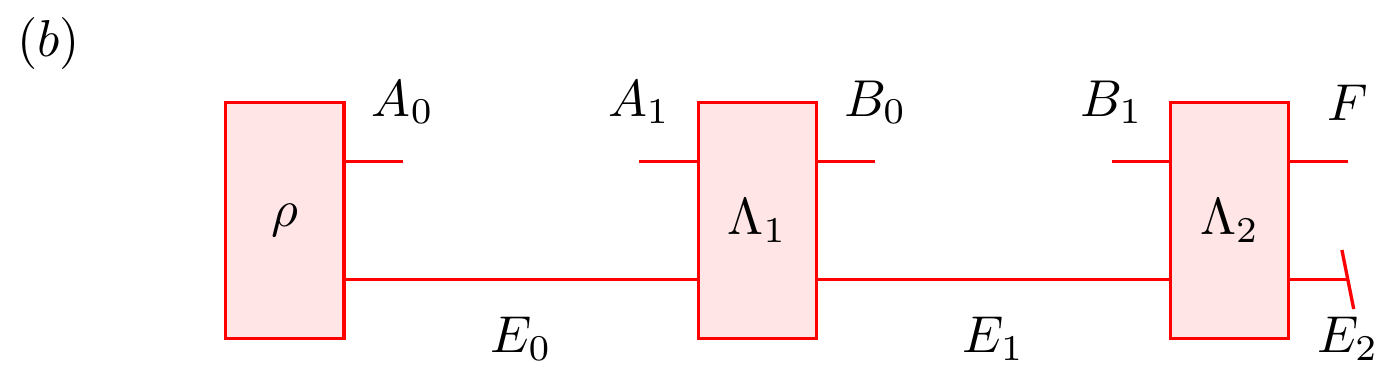} 

\includegraphics[width=0.45\textwidth]{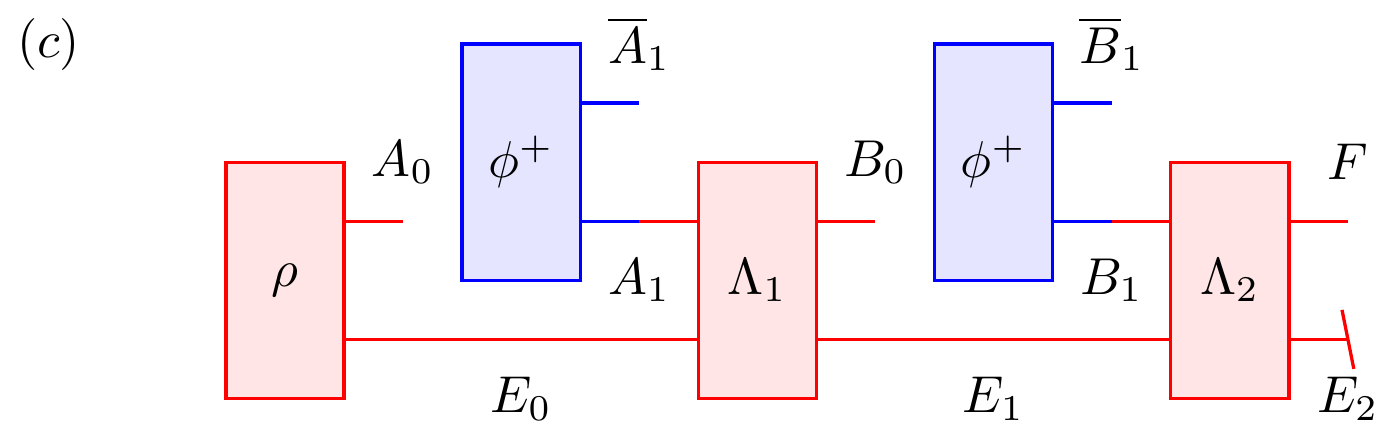}

   \caption{Quantum processes and interventions on non-Markovian processes for the definition of information inequalities. Fig.~(a) shows a quantum process $\Upsilon$ transforms quantum operations $\mathrm{A}$ and $\mathrm{B}$ to a quantum operation from $L(A_0' \otimes P\otimes B_0')$ to $L(A_1' \otimes F\otimes B_1')$, as defined in Eq.~(\ref{quantum_process}). The action of the quantum process $\Upsilon$ on the auxiliary systems $A_0', A_1', B_0', B_1'$ is trivial, in other words on those systems the identity supermap is being applied. Fig.~(b) is a pictorial representation of an arbitrary non-Markovian process: it consists of an initial bipartite system-environment quantum state preparation $\rho$ followed by global quantum transformations $\Lambda_1$ and $\Lambda_2$. The empty spot between the quantum systems $A_0$ and $A_1$ ($B_0$ and $B_1$) represents Alice's (Bob's) local laboratory (Bob). Fig.~(c) shows a diagrammatic representation of the intervention set by Alice and Bob for which the information inequalities in Theorem~\ref{main_thm} hold. Essentially, the interventions are done by respectively feeding the subsystem $A_1$ and $B_1$ of the maximally entangled states $\phi^+_{A_1\bar{A}_1}$ and $\phi^+_{B_1\bar{B}_1}$ in the process. After the interventions, the output is a quantum state (denoted as $\tau$) of the global system $A_0\otimes \bar{A}_1 \otimes B_0\otimes \bar{B}_1 \otimes F$. }
    \label{figure2}
\end{figure}

An $N$-partite quantum process is a multilinear supermap that acts on parts of $N$ quantum channels and results in an output quantum channel. Specifically, let us consider the bipartite case. A bilinear supermap $\Upsilon : T(A_0,A_1) \times T(B_0,B_1) \rightarrow T(P,F)$ is a bipartite quantum process if 
\begin{equation}\label{quantum_process}
 (\mathbf{id}_{A_0' \rightarrow A_1'} \otimes \Upsilon \otimes \mathbf{id}_{B_0' \rightarrow B_1'}) (\mathrm{A},\mathrm{B}) = \mathrm{C}^{\Upsilon, \mathrm{A}, \mathrm{B}}, 
\end{equation}
where $\mathrm{C}^{\Upsilon, \mathrm{A}, \mathrm{B}}$ is a quantum channel in $T(A_0'\otimes P \otimes B_0', A_1'\otimes F \otimes B_1')$ for every pair of quantum channels 
\begin{eqnarray*}
    \mathrm{A}:L(A_0\otimes A_0')\rightarrow L(A_1\otimes A_1'), \\ 
    \mathrm{B}:L(B_0\otimes B_0')\rightarrow L(B_1\otimes B_1').
\end{eqnarray*}
Note that there is no restriction on the dimensions for the quantum systems $A_0',A_1',B_0',B_1'$. In Fig.~\ref{figure2}(a), we show the diagrammatic representation of Eq.~(\ref{quantum_process}). 

The spots for the intervention (shown in the blue shaded regions in Fig.~\ref{figure2}) performed by quantum operations on quantum processes are interpreted as the possible experiments undertaken in local laboratories. The quantum system $A_0$ $(B_0)$ is the input system to the interventions taking place in the laboratory of the part $A$ $(B)$, while the quantum system $A_1$ $(B_1)$ is the corresponding output for the interventions. We also refer here to the parties $A$ and $B$ as Alice and Bob, respectively. Note that, in general, the interventions can be appropriately defined by a collection of CPTNI maps with the sum equal to a CPTP map. 

Interestingly, the theory of multipartite quantum processes is compatible with the existence of processes with an indefinite causal structure. It follows a basic definition concerning processes for which one could reasonably assign a causal order between the local laboratories.

\begin{definition} \label{def_sep} Let $E_j,A_k,B_l$ denote quantum systems with arbitrary dimensions with $j\in \{0,1,2\}$ and $k,l\in \{0,1\}$. Then the following holds:
\begin{enumerate}
    \item A bipartite process $\Upsilon$ is a process with fixed causal order $A \preceq B$ whenever there are a fixed quantum state $\rho \in L(A_0 \otimes E_0)$ and quantum operations $\Lambda_1:L(A_1 \otimes E_0) \rightarrow L(B_0 \otimes E_1)$ and $\Lambda_2:L( B_1 \otimes E_1) \rightarrow L(F \otimes E_2)$ for which holds 
      \begin{multline} \label{non-Markov}
      \Upsilon(\Phi_A,\Phi_B) =  \Tr_{E_2}[(\Lambda_2)(\Phi_B \otimes \mathrm{id}_{E_1})(\Lambda_1)(\Phi_A \otimes \mathrm{id}_{E_0} )(\rho)],
     \end{multline}
for arbitrary quantum channels $\Phi_A : L(A_0) \rightarrow L(A_1)$ and $\Phi_B : L(B_0) \rightarrow L(B_1)$.
   
    \item A process is called causally separable if it is a probabilistic combination of processes with distinct fixed causal orders. Precisely, a process is causally separable and denoted as $\Upsilon_{\mathrm{sep}}$ if there are processes $\Upsilon_{A \preceq B}$ and $\Upsilon_{B \preceq A}$ for which
    \begin{equation}
    \Upsilon_{\mathrm{sep}}= q\Upsilon_{A \preceq B}+(1-q)\Upsilon_{B \preceq A},
    \end{equation}
    with probability $0 \leq q\leq 1$.
\end{enumerate}
\end{definition}

Fig.~\ref{figure2}(b) provides an illustration of a fixed causal order in agreement with Definition~\ref{def_sep}. While the fixed causal ordered processes can include Markovian processes as well, the most general fixed causal ordered processes correspond to non-Markovian processes. Hence, from now on, we use the term `non-Markovian' and `fixed causal ordered' processes interchangeably.  When we say a process is `beyond non-Markovian', we refer to causally separable and indefinite causal ordered processes. In particular, the most general fixed causal order processes are implemented by a bipartite system-environment time-evolution of an initial quantum state, allowing the local laboratories to implement interventions on the system at fixed different instants of time, and furthermore, are quantum processes with a trivial global past system. In that sense, multipartite quantum processes of the non-Markovian form are also called multitime quantum stochastic processes. The reader is referred to Ref.~\cite{milz2021quantum} for a recent review on the topic, and Refs.~\cite{Giarmatzi_2021,Nery_2021,Taranto_2024,goswami2024hamiltonian} for different types of non-Markovianity. In turn, causally separable processes are probabilistic mixtures of quantum non-Markovian processes with different causal orders. The reader is also urged to follow Ref.~\cite{branciard2016witnesses} for a detailed discussion on the role of separability in the process formalism.

In what follows, it is very convenient to consider the representation of quantum processes provided by the process matrix formalism. To every bipartite quantum process $\Upsilon$ with local laboratories $A$ and $B$ it is associated a process matrix $W$ in $L(P  \otimes A_0\otimes A_1\otimes B_0 \otimes B_1 \otimes F)$. The process matrix is defined in such a way that for arbitrary quantum operations $\mathrm{A}$ and $\mathrm{B}$, we have \cite{oreshkov2012quantum}
\begin{multline}
J(\Upsilon(\mathrm{A},\mathrm{B}))= \Tr_{AB}[W^{T_{AB}}(\mathds{1}_{P}\otimes J(\mathrm{A})\otimes J(\mathrm{B}) \otimes \mathds{1}_{F})],
\end{multline}
where $T_{AB}$ denotes the partial transposition with respect to the systems $A_0,A_1,B_0,B_1$. In order for a given process matrix to represent a valid quantum process, it must respect conditions which can be found in Ref.~\cite{giarmatzi2021witnessing}. A quantum process $\Upsilon$ with rank-1 process matrix $W=\proj{w}$ is called a pure process, and $\ket{w}$ is called the process vector.

\subsection{Information inequalities}
Data processing inequalities stand as crucial information limitations, providing the mathematical tools involved in the proof of coding theorems since the early stages of classical and quantum Shannon theories. For instance, the converse statement of the channel-coding theorem, reported in the seminal work of Claude E. Shannon in Ref.~\cite{shannon1948mathematical}, can be proved with the use of data processing inequalities \cite{yeung2008information}: there is no asymptotic block encoding-decoding scheme of a memoryless channel for which the communication rate exceeds the channel capacity with arbitrarily small error probability. Similarly, the goal of this study is to provide further understanding of the information limitations in the realm of quantum processing.

Classical data processing inequalities hold as follows. Let $X_1 \rightarrow X_2 \rightarrow X_3$ denote a classical Markovian process, that is, a collection of classical random variables for which the probability distribution of $X_3$ might depend conditionally upon the variable $X_2$, nevertheless, it does not depend upon $X_1$. Intuitively, one can understand such a process as the observation of a quantity $X$ in three discrete and different instances, in such a way that the information on the future observation $X_3$ is fully contained in the present $X_2$, while being independent of its past state $X_1$. Furthermore, it holds the information inequalities~\cite{yeung2008information}
\begin{eqnarray}
    I(X_1;X_2) &\geq&  I(X_1;X_3), \label{cdpi1}\\
    I(X_2;X_3) &\geq&  I(X_1;X_3), \label{cdpi2}
\end{eqnarray}
where $I(X_r ; X_s)$ denotes the mutual information of random variables $X_r$ and $X_s$, with $r,s=1,2,3$. It is worth noting that for longer classical Markovian processes, there are classical information inequalities which are not equivalent to the couple above. We refer the reader to Ref.~\cite{capela2020monogamy,capela2022quantum} -- and the references therein -- in order to find detailed information on classical data processing inequalities.  

Data processing inequalities are useful witnesses on non-Markovian behavior from marginal information as well. That is, in order to assert a given classical process drawn from a probability distribution $p(X_1,X_2,X_3)$ is non-Markovian, one has to check that the Markov condition fails: $p(X_3|X_1,X_2) \neq p(X_3|X_2)$. Furthermore, it shows that the knowledge of the joint probability distribution of the random variables $X_1,X_2,X_3$ is necessary. Nevertheless, the data processing inequalities in Eqs.~(\ref{cdpi1},\ref{cdpi2}) work as witnesses on non-Markovianity from limited knowledge on the process' probability distribution: it might still be possible to check non-Markovianity from the violation of information inequalities with respect to the marginal probability distribution of pairs of variables of the process only. We refer the reader to Ref.~\cite{capela2020monogamy} for a detailed discussion on this application of data processing inequalities.

Now, the limitations imposed by quantum processes take place as central results in the theory of quantum information \cite{watrous2018theory}. For instance, the quantum data-processing inequality in terms of quantum coherent information reported in Ref.~\cite{schumacher1996quantum} is a cornerstone in quantum information theory. Importantly, it follows from such a consideration the necessary and sufficient conditions on the input message, represented with a quantum state, and on the noisy quantum channel through which quantum information goes, for which the receiver can perfectly error-correct the output signal. Quantum data processing inequalities can also be formulated with respect to distinct information measures, such as the quantum mutual information, as it can be found in Refs.~\cite{buscemi2014complete,hayden2004structure}. 

Quantum data processing inequalities can also serve as important witnesses on quantum non-Markovian phenomena \cite{capela2022quantum}, that is, the multi-time correlations emerging from system-environment dynamics \cite{pollock2018non}. Similarly to their classical counterpart, it is possible with the assistance of quantum information inequalities to witness quantum non-Markovianity with interventions on the system-part only. 

Although system-environment dynamics is a powerful model allowing for the description and implementation of many interesting quantum mechanical effects \cite{taranto2020exponential,taranto2021non}, recent studies on quantum information theory have culminated in a novel perspective on how to deal with the possible quantum processes \cite{chiribella2008transforming,chiribella2013quantum}, and importantly, predicting the existence of valid process in nature beyond this model \cite{oreshkov2012quantum}. Inspired by the interventional approach, in such a way that quantum experiments performed by local agents take place as probabilistic quantum operations called quantum instruments, quantum processes are represented as higher-order quantum operations \cite{oreshkov2012quantum,araujo2015witnessing}. Under a similar framework, a novel platform for the study of quantum non-Markovian phenomena has emerged, allowing for further investigation and understanding in the area \cite{milz2019completely,milz2020kolmogorov,taranto2019quantum}. We refer the reader to Ref.~\cite{milz2021quantum} for a review on the interventional perspective on quantum non-Markovianity. 

Importantly, quantum processes in a scenario with several parties can exist in such a way that they are not just a conventional time-ordered quantum circuit connecting the operations implemented by the local agents \cite{oreshkov2012quantum}. In other terms, there are quantum processes which go beyond the non-Markovian model, and are termed as processes with indefinite causal order. Furthermore, for multi-party quantum processes, non-Markovianity is not the most extreme case. 

Then, what would be the data processing conditions emerging from non-Markovianity? Would the corresponding information inequalities be sharp enough in order to certify that a given process is beyond non-Markovianity? Now we are ready to consider the main question addressed here: given a bipartite process $\Upsilon$, how can one certify it is not compatible with a fixed causal structure by defining an intervention on the accessible quantum systems $A_0,A_1$ and $B_0,B_1$? The problem, which can be easily generalized to a multi-party scenario, is addressed in the following section.

\bigskip

\section{Results} 

\label{sec_witnesses}

The goal of this section is to set up methods to determine whether a given quantum process is not compatible with a fixed causal structure. Let us consider, in the subsequent discussion, without loss of generality, that $\Upsilon$ is a bipartite quantum process with global future $F$. Furthermore, we aim to determine that $\Upsilon$ is neither of the form $A \preceq B$ nor  $B \preceq A$. 

\subsection{Information inequalities for non-Markovian quantum processes}

A possible way to address such a problem is to design data-processing conditions as similarly done in Refs.~\cite{capela2020monogamy,capela2022quantum}. Information inequalities work as necessary conditions on processes of a particular form, and thus, a violation of them asserts that the given process is not of the form considered. In this sense, the following information conditions are proposed in order to witness processes beyond non-Markovianity in the more general setting proposed by the process formalism.

\begin{theorem}\label{main_thm} Let  $\Upsilon_{A \preceq B}$ be a bipartite non-Markovian process with global future $F$.
Take the quantum state preparation $\tau$ for the joint system $A_0\otimes A_1\otimes B_0\otimes B_1\otimes F$ as the outcome from the intervention on the process defined by feeding each of the systems $A_1$ and $B_1$ with half of a maximally entangled state (see Fig.~\ref{figure2}(c) for a pictorial representation). Then, it holds the information inequality
\begin{equation}\label{ineq_ent}
    H(A_0A_1B_0B_1F)_{\tau} -H(A_0A_1B_0)_{\tau} \geq  \log_2 \frac{\dim(B_1)}{\dim(F)},
\end{equation}
where the entropy of a generic quantum system X prepared in the state $\rho$ is defined with $H(X)_{\rho}=-\Tr [\rho \log_2 \rho]$.
\end{theorem}

Furthermore, it follows as a simple consequence that whenever the dimension of Bob's output is not smaller than the dimension of the global future, i.e., $\dim (B_1)\ge \dim(F)$, with respect to a non-Markovian process of the form $A \preceq B$, then, the following holds:
\begin{equation} \label{entropy_increase}
    H(A_0A_1B_0B_1F)_{\tau} \geq H(A_0A_1B_0)_{\tau},
\end{equation}
where $\tau$ has been defined in the statement of Theorem~\ref{main_thm}.

The proof of the result above can be found in detail in Appendix \ref{appendix_proof}. Importantly, the crucial application behind it is that one can build up witnesses that a given quantum process is not in the set of non-Markovian ones. For that sake, let $\Upsilon$ denote a bipartite quantum process with global future $F$ and local agents $A$ and $B$. Thus, in case the process under consideration is non-Markovian with the part $A$ in the causal past of $B$, then it holds that:
\begin{align}
     \textrm{DP}_{\Upsilon}^{A\preceq B} &\coloneqq  H(A_0A_1B_0B_1F)_{\tau} - H(A_0A_1B_0)_{\tau} \label{dpAB1} \\
    &\geq \log_2 \frac{\dim(B_1)}{\dim(F)}. \label{dpAB2}
\end{align}
The lower bound stated above shows that the quantity defined in Eq.~(\ref{dpAB1}) is a witness on indefinite causal order: a quantum process $\Upsilon$ with $\textrm{DP}$ smaller than $\log_2 \dim(B_1) - \log_2 \dim(F)$ necessarily violates the condition $A \preceq B$.

Now, it is clear that interchanging the roles of the local agents $A$ and $B$ in Theorem~\ref{main_thm}, we have similar conclusions for non-Markovian processes with $B$ in the local past of $A$. That is, it follows that
we must have the validity of the data-processing condition
\begin{equation}\label{ineq_ent_2}  H(A_0A_1B_0B_1F)_{\tau} -H(B_0B_1A_0)_{\tau} \geq \log_2 \frac{\dim(A_1)}{\dim(F)}.
\end{equation}
Note that the interventional approach adopted in the preparation of $\tau$ is independent of the causal order between the local agents.

In turn, in case $\Upsilon$ is a non-Markovian process of the form $B \preceq A$, it holds that
\begin{align}
     \textrm{DP}_{\Upsilon}^{B\preceq A} &\coloneqq  H(A_0A_1B_0B_1F)_{\tau} - H(B_0B_1A_0)_{\tau} \label{dpBA1} \\
    & \geq  \log_2 \frac{\dim(A_1)}{\dim(F)}. \label{dpBA2}
\end{align}

Importantly, a violation of both conditions expressed in Eqs.~\eqref{dpAB2}, and \eqref{dpBA2} proves the corresponding quantum process is not compatible with a fixed causal structure.

\begin{figure}[h!] 
\centering   

\includegraphics[width=0.45\textwidth]{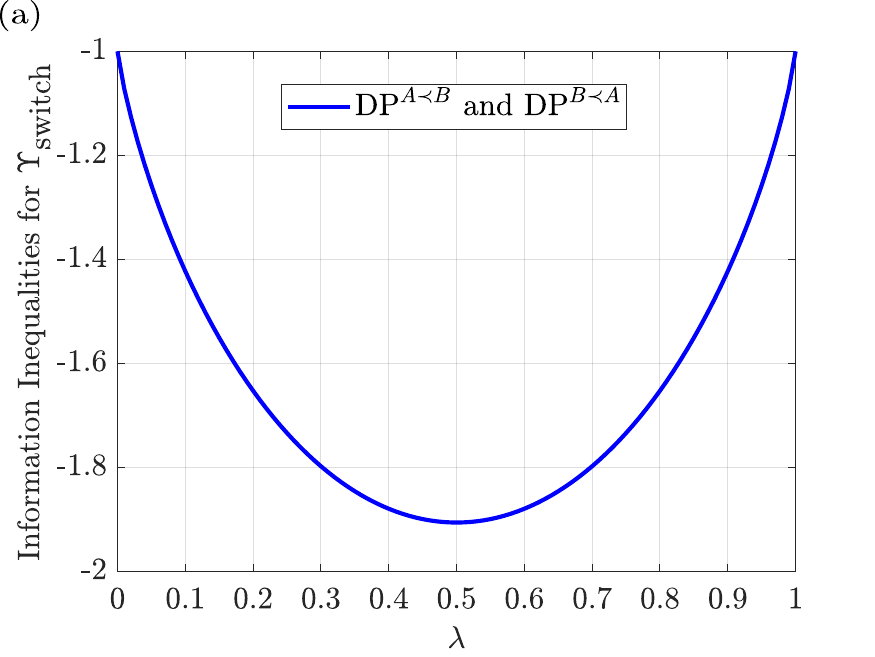} 

\includegraphics[width=0.45\textwidth]{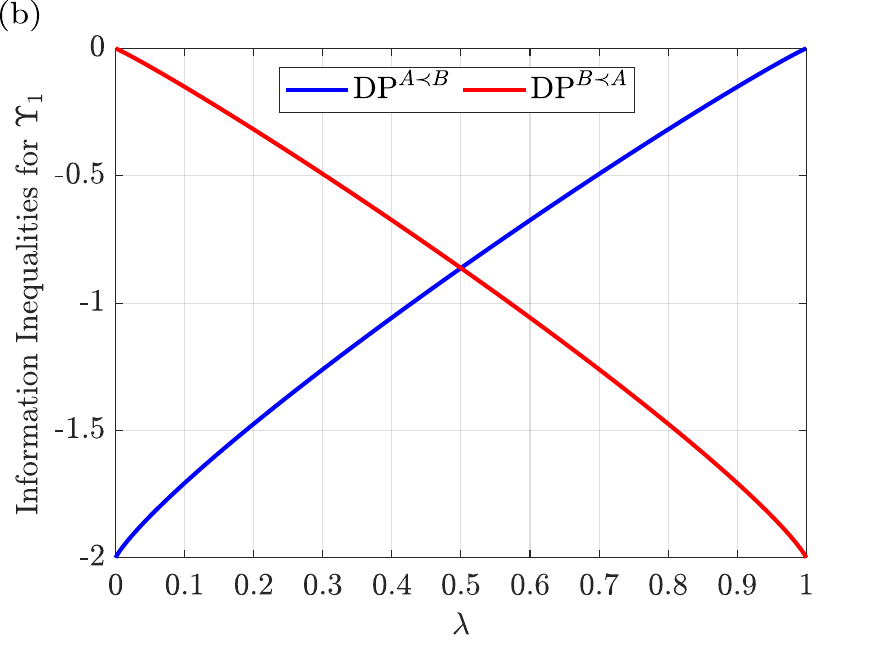} 

\includegraphics[width=0.45\textwidth]{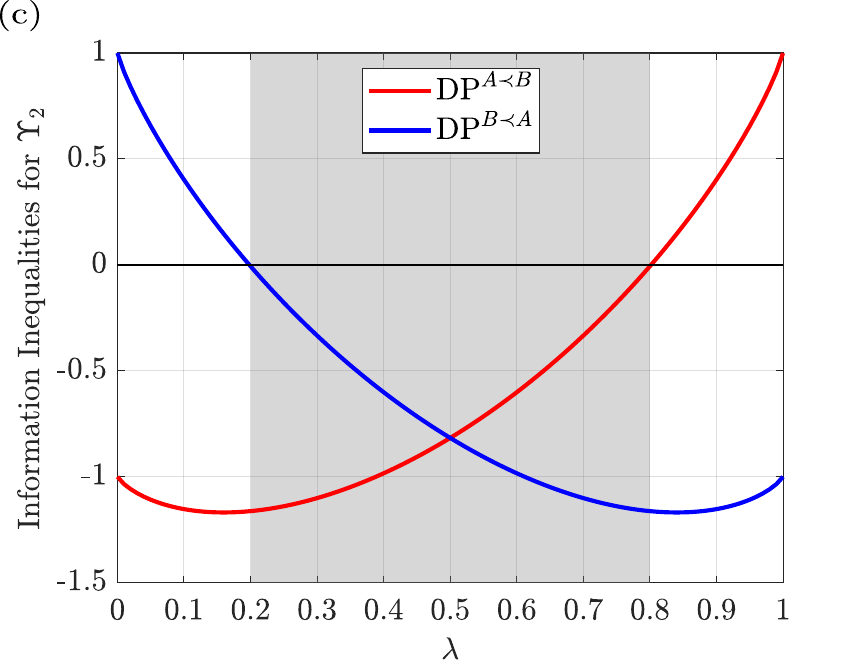} 

    \caption{Violation of the information inequalities for the quantum switch: the plots show the witnesses in Eqs.~(\ref{dpAB1},\ref{dpAB2},\ref{dpBA1},\ref{dpBA2}) as a function of the input control system state parameter $\lambda$. Plot (a) corresponds to the quantum switch (with target and control). Plot (b) considers the quantum switch with output control system $C_1$ traced out ($\Upsilon_1$ defined in Eq.~(\ref{separable_process})), while plot (c) considers the output target system $T_1$ traced out ($\Upsilon_2$ Here, the information inequalities witness that the process is beyond non-Markovianity only in the shaded region.}
    \label{plot_1}
\end{figure}

\subsection{Violation of the information inequalities from processes with indefinite causal order}

Now that we have set up the necessary conditions on quantum processes with a fixed causal order, we are ready to test how efficiently those quantities are in witnessing processes not compatible with non-Markovianity. The case study addressed here is the coherent superposition of causal orders named as a quantum switch in the literature.

\textit{The quantum switch and causally separable processes.---}The quantum switch is a quantum process $\Upsilon_{\textrm{switch}}$ with bipartite global past and future. Its global past is constituted by an input target system denoted with $T_0$ and an input control system $C_0$, and its global future is defined with the output target system $T_1$ and the output control system $C_1$. The action of quantum switch $\Upsilon_{\textrm{switch}}$ on the channels $A: L(A_0)\to L(A_1)$, with Kraus operators $\{\mathsf{A}_i\}_i$ and $B: L(B_0)\to L(B_1)$, with Kraus operators $\{\mathsf{B}_i\}_i$ results in a new channel $\Upsilon_{\textrm{switch}}(A,B): L(T_0 \otimes C_0)\to L(T_1 \otimes C_1)$, with Kraus operators $\{\upsilon_{ij}\}_{ij}$, such that for all quantum states $\rho \in L(T_o)$ and $\varphi \in L(C_0)$, we have
\begin{align}
&\Upsilon_{\textrm{switch}}(A,B)(\rho \otimes \varphi) = \sum_{i,j}\upsilon_{ij} (\rho \otimes \varphi) \upsilon_{ij}^\dagger \ \textrm{where} \label{switch} \ \\
&\upsilon_{ij} \coloneqq \mathsf{B}_i\mathsf{A}_j \otimes \proj{0} + \mathsf{A}_j\mathsf{B}_i \otimes \proj{1}.
\end{align}
Note that discarding the output control system $C_1$ of the quantum switch $\Upsilon_{\textrm{switch}}$ results in a classical mixture of the fixed causal ordered processes  $\Upsilon_1$ such that
\begin{eqnarray}\label{separable_process}
 \Upsilon_1(\mathrm{A},\mathrm{B})&\coloneqq&
 \Tr_{C_1}[\Upsilon_{\textrm{switch}}(\mathrm{A},\mathrm{B})(\rho \otimes \varphi) ]\\ &=& \lambda (\mathrm{B} \circ \mathrm{A}) (\rho)+(1-\lambda)(\mathrm{A} \circ \mathrm{B}) (\rho),
 \end{eqnarray}
for every input target system prepared in the state $\rho$, and the input control system $\varphi:=\proj{\varphi}$ is prepared in the pure state $\ket{\varphi}=\sqrt{\lambda} \ket{0}+\sqrt{1-\lambda} \ket{1}$.

Fig.~\ref{plot_1} presents plots for the witnesses defined in Eqs.~(\ref{dpAB1},\ref{dpBA1}) for three distinct processes with input target $T_0$ system prepared in the pure state $\ket{0}$: (a) the quantum switch defined in Eq.~(\ref{switch}), (b) the separable process defined in Eq.~(\ref{separable_process}), and (c) the quantum switch with the input control system prepared in the quantum state $\ket{\varphi}$ and the output target system $T_1$ discarded:
\begin{equation} \label{switch_C}
    \Upsilon_2(\mathrm{A},\mathrm{B}) \coloneqq 
 \Tr_{T_1}[\Upsilon_{\textrm{switch}}(\mathrm{A},\mathrm{B})(\rho \otimes \varphi) ].
\end{equation}
Note that the local dimensions of the processes in Fig.~\ref{plot_1}(b) and Fig.~\ref{plot_1}(c) are in agreement with the conditions for which Eq.~(\ref{entropy_increase}) holds. 

Importantly, the plots in Fig.~\ref{plot_1}(a) and Fig.~\ref{plot_1}(b) provide a successful certification that the processes $\Upsilon_{\textrm{switch}}$ and $\Upsilon_1$ are neither of the form $A \preceq B$ nor $B \preceq A$, respectively, for all values $0<\lambda<1$; for $\lambda=0,1$ the processes are indeed Markovian, and no violation of the information principles are expected to take place. For the process $\Upsilon_2$, the witnesses are able to prove it as incompatible with a fixed causal order for the range of parameters $0.2 \leq \lambda \leq 0.8$. It is worth to note that no conclusion can be made about the causal structure of $\Upsilon_2$ for the values of $\lambda$ for which are found no violation of the information inequalities.

\begin{figure}[h!] 
\centering   

\includegraphics[width=0.5\textwidth]{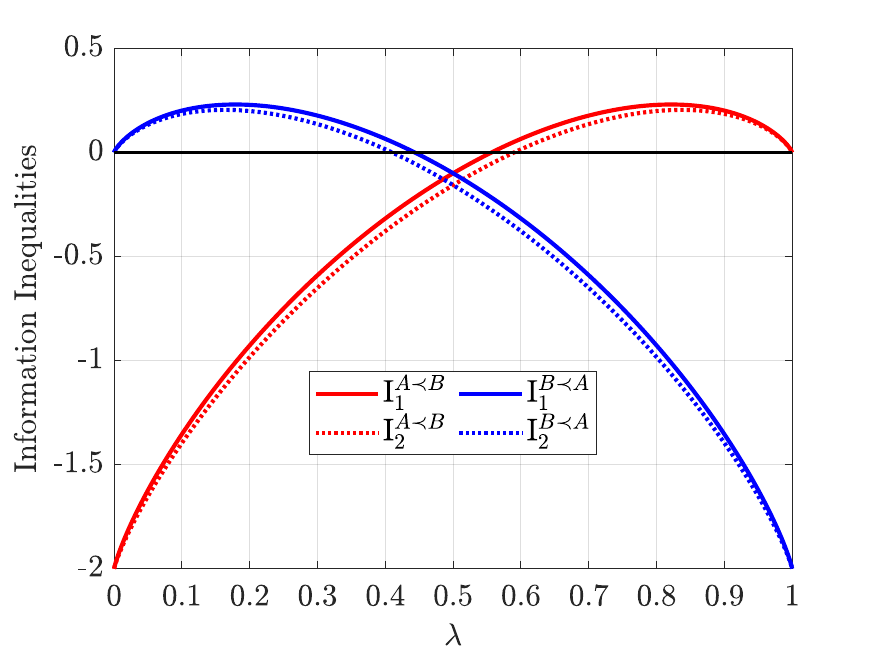}

    \caption{Marginal witnesses $\mathrm{I}_1$ and $\mathrm{I}_2$ for the process with output control system traced out ( $\Upsilon_1$ as in Eq.~\eqref{separable_process}). The plot presents the marginal witnesses for $A \preceq B$ and $B \preceq A$ with respect to the parameter $\lambda$ defining the quantum state preparation $\sqrt{\lambda}\ket{0} + \sqrt{1-\lambda}\ket{1}$ of the initial control system $C_0$. Interestingly, for this particular example, the witness $\mathrm{I_2}$ (dashed lines) provides slightly better results than $\mathrm{I_1}$ (solid lines). A negative value for $\mathrm{I_1}$ or $\mathrm{I_2}$ shows that the quantum separable process $\Upsilon_1$ is not conceivable with a fixed causal structure.}
    \label{plot_marginal}
\end{figure}

\subsection{Data-processing inequalities on non-Markovianity from marginal information} \label{sec_4}

The development of data-processing conditions is ubiquitous in classical and quantum information theory. In particular, it allows us to impose limitations on information processing from the knowledge of marginal distributions only. For the sake of clarity, let us consider a classical three-time-step Markovian process denoted with $X_1 \rightarrow X_2 \rightarrow X_3$. Then, although a complete characterization of the process would necessarily involve the joint probability distribution of the three random variables, bounds on the possible temporal correlations take place in the form of data-processing inequalities. In particular, the marginal distributions for the bipartitions must hold $I(X_1;X_2) \geq I(X_1;X_3)$ and $I(X_2;X_3) \geq I(X_1;X_3)$, where $I(X_i;X_j)$ generically denotes the mutual information of random variables $X_i$ and $X_j$. Furthermore, with access to marginal probability distributions, it is indeed possible to check whether a classical process is non-Markovian given a violation of the inequalities above.

Here, we consider a similar reasoning for the problem addressed: given a quantum process $\Upsilon$ with local laboratories with input and output systems $A_0,B_0$ and $A_1,B_1$, respectively, and global future $F$, can we assign relevant limitations on non-Markovianity from marginal information? We answer this question affirmatively by considering, as an example, the causally separable quantum processes $\Upsilon_1$ defined in Eq.~(\ref{separable_process}). 

In order to define marginal conditions on non-Markovianity, we combine the entropic inequality in Eq.~(\ref{ineq_ent}) along with the strong subadditivity condition for quantum entropy, that is, a tripartite quantum system $X \otimes Y \otimes Z$ must satisfy, regardless of its global quantum state,
\begin{equation} \label{sse}
    H(XY)+H(YZ) \geq H(XYZ)+H(Y).
\end{equation}

Now, let $\Upsilon$ be a bipartite non-Markovian quantum process with local laboratories $A$ and $B$ and global future F. Assume also that Alice and Bob's local interventions satisfy the same conditions stated in Theorem~\ref{main_thm}. Then, in case $A \preceq B$ holds, it follows that the information quantities
\begin{eqnarray}
    \mathrm{I}_1^{A \preceq B} \coloneqq H(B_1|A_0A_1B_0)+H(F|A_1B_1), \\
    \mathrm{I}_2^{A \preceq B} \coloneqq H(A_0F|A_1B_1)+H(B_0|A_1B_1),
\end{eqnarray}
are not smaller than $\log_2 \dim(B_1) - \log_2 \dim(F)$,
where $H(Y|X) \coloneqq H(XY)-H(X)$ is the conditional quantum entropy. Note that the quantum version of conditional entropy may be negative, and therefore, the quantities above are non-trivial witnesses on indefinite causal order. 

Otherwise, in case $B \preceq A$ is valid, we have interchanged the roles of Alice and Bob, so that the information quantities are defined as
\begin{eqnarray}
    \mathrm{I}_1^{B \preceq A} \coloneqq H(A_1|A_0B_0B_1)+H(F|A_1B_1), \\
    \mathrm{I}_2^{B \preceq A} \coloneqq H(B_0F|A_1B_1)+H(A_0|A_1B_1),
\end{eqnarray}
are greater or equal to $\log_2 \dim(A_1) - \log_2 \dim(F)$.

It can be demonstrated that the quantities $\mathrm{I}_1$ and $\mathrm{I}_2$ satisfy the above-mentioned lower bounds with assistance of the elemental condition in Eq.~(\ref{sse}) and the data-processing condition in Theorem~\ref{main_thm}. Let us consider the case $A \preceq B$. Then, $\mathrm{I}_1$ follows from the assignment $X=A_0 \otimes B_0,Y=A_1 \otimes B_1,Z=F$, while $\mathrm{I}_2$ follows from $X=A_0 \otimes F$, $Y=A_1 \otimes B_1$, $Z=B_0$. It is important to note that for the computation of those quantities, it is not necessary to have access to the global interventional state $\tau$, and for that reason, the novel information quantities are called marginal witnesses. Figure~\ref{plot_marginal} shows how the marginal witnesses perform on the certification of indefinite causal order of the quantum processes $\Upsilon_1$.

It is important to note that since strong subadditivity of quantum entropy holds for every state, the marginal witnesses are never more efficient than the information witness defined in Eq.~(\ref{dpAB1}). That can indeed be visualised by comparing the plot in Fig.~\ref{plot_marginal} and the middle picture in Fig.~\ref{plot_1}. Nevertheless, the operational advantage of the marginal witnesses is clear from their definition from subsystems only.

\begin{figure}[t] 
\centering   
\includegraphics[width=0.5\textwidth]{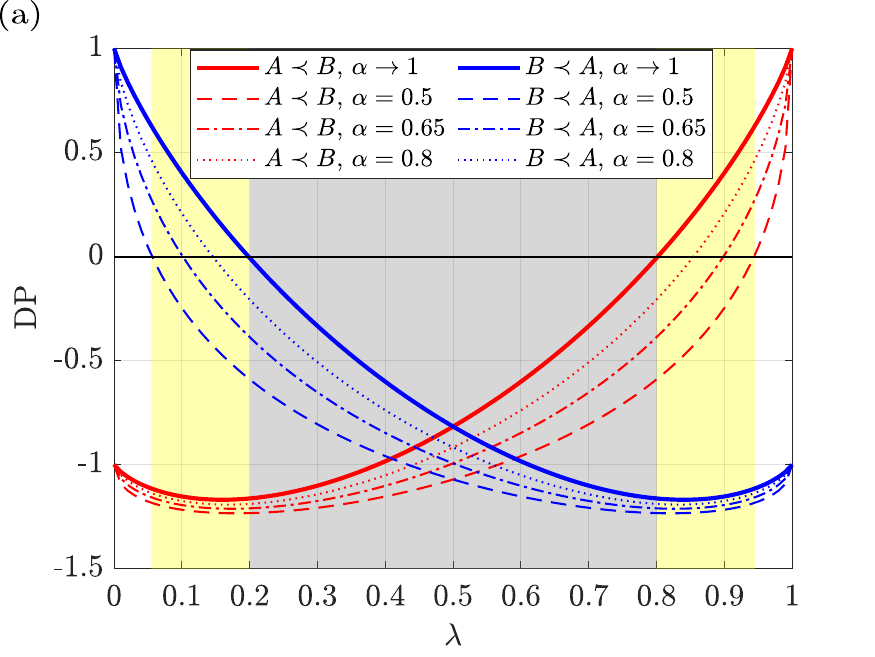} 
\includegraphics[width=0.5\textwidth]{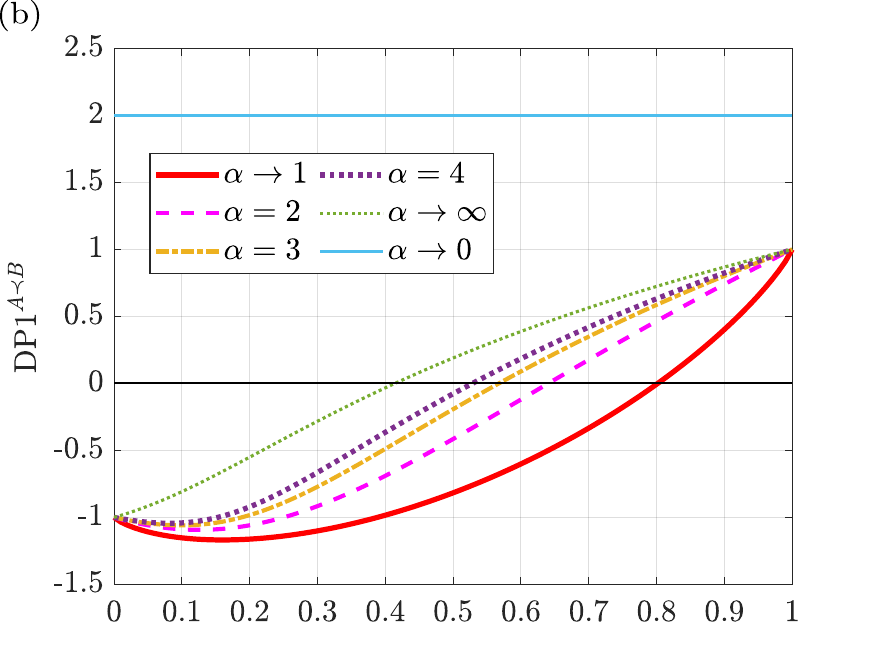} 
    \caption{Rényi information inequalities $\textrm{DP}$ for the quantum switch with output target system discarded ($\Upsilon_2$ as in Eq.~\eqref{switch_C}) as a function of $\lambda$. Plot (a) shows how $\textrm{DP1}$ for $\alpha=0.5,0.65,0.8$. Interestingly, we have an advantage for the region of $\lambda$ for which the certification of classical mixture of causal orders holds (yellow area in the plot). Plot (b) shows that for the values larger than one ($\alpha=2,3,4$ and $\alpha \rightarrow \infty$), the Rényi information inequalities are less sensitive than the one for von Neumann quantum entropy.}
    \label{plot_3}
\end{figure}

\subsection{Witnessing processes beyond non-Markovianity with generalized quantum entropies} \label{sec_5}

The quantum entropy $H(A)_{\rho}$ is a ubiquitous information measure in quantum information theory. For instance, it sets the limit for the compression of a quantum source with state $\rho$ \cite{schumacher1995quantum}. Moreover, information measures related to the quantum entropy are also related to conditions on single-shot quantum error-correction \cite{schumacher1996quantum}. We refer the reader to Refs.~\cite{watrous2018theory,wilde2011classical} for a comprehensive presentation on quantum entropy and its use in quantum information theory.

In spite of the fundamental relevance of quantum entropy in information theory, with respect the goals considered in this study---the certification of processes beyond non-Markovianity,---other information measures could as well provide useful witnesses.

Now, we focus on quantum processes $A \preceq B$ for which the global future is such that $\dim(F)$ is not larger than $\dim B_1$. Similar reasoning holds for processes of the form $A \preceq B$, interchanging $A$ and $B$. The properties of quantum entropy necessary for the validity of Eq.~(\ref{entropy_increase}), as it can be found in Appendix~\ref{appendix_proof}, are as follows:
\begin{enumerate}[(i)]
    \item Monotonicity under the action of quantum operations preserving maximally mixed states: the condition $H(B)_{\Lambda(\rho)} \geq H(A)_{\rho}$ holds as long as $\mathrm{dim}(B) \geq \mathrm{dim}(A)$ and $\Lambda(\omega_A)=\omega_B$, where $\omega_X$ denotes the maximally-mixed state of $X$.
    \item It holds $H(A)_{\psi}=H(B)_{\psi}$ for any bipartite pure state $\psi$ of the quantum system $A \otimes B$.
\end{enumerate}

Furthermore, any information quantity with the properties above could be used in the definition of the information inequalities developed here. In particular, it is worth considering the quantum Rényi entropies, defined for a quantum system $A$ in the state $\rho$ with
\begin{equation}\label{Renyi}
    H_{\alpha}(A)_{\rho}=\frac{1}{1-\alpha}\log_2 \Tr[\rho^{\alpha}],
\end{equation}
for parameters $0 < \alpha < 1$ and $1 < \alpha < \infty$. Importantly, the conditions (i) and (ii) hold for the range of parameter $\frac{1}{2} < \alpha < 1$ and $1 < \alpha < \infty$ \cite{frank2013monotonicity}.

In this section, we show how Rényi entropies might provide advantages in witnessing classical mixture of distinct orders by considering the particular quantum process $\Upsilon_2$, a situation for which the information inequality in Eq.~(\ref{entropy_increase}) fails to perfectly detect it. 

Note that we recover the quantum entropy $H(A)_{\rho}$ from Eq.~(\ref{Renyi}) as $\alpha \rightarrow 1$. On the other hand, the limiting case $\alpha \rightarrow 0$ defines the max-entropy
\begin{equation}
    H_{\mathrm{max}}(A)_{\rho}=\log_2\rank(\rho),
\end{equation}
while the case $\alpha \rightarrow \infty$ defines the min-entropy
\begin{equation}
     H_{\mathrm{min}}(A)_{\rho}=-\log_2 \lVert \rho \rVert,
\end{equation}
where $\lVert \rho \rVert$ denotes the operator norm of $\rho$; the operator norm of a positive-semidefinite operator is equal to its maximum eigenvalue. The max- and min-entropies are also examples of quantum entropies respecting the desired properties \cite{datta2009min}.

The plots in Fig.~\ref{plot_3} shows the values of $\mathrm{DP}^{A \preceq B; \alpha} \coloneqq H_{\alpha}(A_0A_1B_0B_1F)_{\tau} -H_{\alpha}(A_0A_1B_0)_{\tau}$ for different initial control states, that is, for different values of $\lambda$. The quantity  $\mathrm{DP}^{A \preceq B; \alpha}$ is defined similarly replacing $A$ with $B$.  
It is worth mentioning that our goal here is not to be exhaustive with the possibilities of information measures to be addressed, but just to consider alternatives to the traditional quantum entropy.

In that sense, the picture in Fig.~\ref{plot_3}(a) shows that the Rényi information inequalities for the process $\Upsilon_2$, defined in Eq.~(\ref{switch_C}) as a function of the initial control state parameter $\lambda$, presents advantage over the quantum entropy with respect to different values of $\frac{1}{2} \leq \alpha < 1$. On the other hand, the bottom picture in Fig.~\ref{plot_3}(b) shows that the Rényi information inequalities for particular values of $\alpha > 1$ and the min- and max-entropy are less sensitive to indefinite causal order than the information conditions with respect to the ordinary quantum entropy. 

\section{Conclusion} \label{sec_6}
In this article, we propose entropic inequalities which are respected by every fixed causal ordered process, that is, processes represented by quantum combs. The violation of the inequalities is sufficient to certify that the underlying quantum process is incompatible with a fixed causal order---the process is either causally separable or has indefinite causal order. To derive the inequalities, we adopt an interventional approach performed by the local agents of the process.

We have used such a condition to certify a well-known quantum process beyond fixed causal orders: the quantum switch. Taking either the output target or control system as a global future, we were able to certify its departure from a fixed causal order, as discussed in Sec.~\ref{sec_witnesses}. 

Additionally, we also show that strong subadditivity of quantum entropy, combined with the main result in this manuscript, results in a new class of witnesses associated with the marginal process. These inequalities are relevant in the scenario where the process is only partially accessible to the local agents. Finally, we show that similar entropic inequalities hold for $\alpha$-Rényi entropies for a wide range of $\alpha$. We find that, in the range of $\alpha\in [1/2,1)$, the Rényi entropic inequalities are advantageous in the certification of processes beyond fixed causal orders.
 
For future works, the current result could be generalized to multi-partite non-Markovian processes involving other classes of entropic inequality. Furthermore, it will be interesting to develop an operational interpretation of our entropic inequality in Theorem.~\ref{main_thm}, similar to the operational meaning \cite{shannon1948mathematical} of standard data processing inequalities. We refer the interested reader to Refs.~\cite{schumacher1996quantum,hayden2004structure,buscemi2014complete} for a presentation of operational interpretations of data-processing theorems in quantum information theory. 

\section*{Acknowledgments} 
\noindent
This study was financed, in part, by the São Paulo Research Foundation (FAPESP), Brasil. Process Number 2024/06972-9 and Process Number 2022/00209-6. K.G. is supported by the Hong Kong
Research Grant Council (RGC) through grant No. 17307520, John Templeton Foundation through grant 62312, “The Quantum Information Structure of Spacetime” (qiss.fr).

\appendix

\section{Proof of the information condition in Eq.~(\ref{ineq_ent}) for processes with fixed causal order}\label{appendix_proof}

Theorem~\ref{main_thm} depends on the fact that feeding the input system of local laboratories of non-Markovian quantum processes with maximally mixed quantum states forces their environment to undergo an evolution preserving maximally mixed states for all time-steps. In the following, we describe it in detail. Let us first consider a well-known result in the literature upon which its proof is built.

\begin{lemma} \label{lemma_dpi} Let $\Lambda:L(A)\rightarrow L(B)$ be a quantum channel mapping maximally mixed states to maximally mixed states: $\Lambda(\omega_A)=\omega_B$ where $\omega_X=\mathds{1}_X/\dim(X)$ is a maximally mixed state in the system $X$. Then, it holds that
\begin{equation}
    H(B)_{\Lambda(\rho)}-H(A)_{\rho} \geq \log_2 \frac{\dim(B)}{\dim(A)}.
\end{equation}
\end{lemma}

\begin{proof} The quantum relative entropy is a monotonically decreasing function under the action of quantum operations defined as \cite{watrous2018theory}
\begin{equation}
    D(\rho||\sigma)= \begin{cases}
    \Tr[\rho(\log_2 \rho - \log_2 \sigma)], &\text{if $\rho \subseteq \mathrm{supp}(\sigma)$ }\\
    +\infty, &\text{otherwise}. 
    \end{cases}
\end{equation}

Now, for channels preserving maximally mixed state, $\Lambda(\omega_A)=\omega_B$, monotonicity of relative entropy implies
    \begin{equation} \label{MRE}
    D(\rho||\omega_A) \geq D(\Lambda(\rho)||\omega_B).
    \end{equation}
    
Given any quantum state $\rho \in L(X)$, we have $\rho \subseteq \mathrm{supp}(\omega_X)$, and the following holds
\begin{align}
   D(\rho || \omega_X) &= \Tr \left[ \rho ( \log_2(\rho) - \rho \log_2(\omega_X) ) \right] \nonumber \\
   &= -H(X)_{\rho} + \log_2 ( \dim(X) ). \label{Eq:relative_with_max_mixed}
\end{align}
In the last equality, we have used $\log_2(\omega_X)=-\log_2(\dim(X))\mathds{1}_X$. We now prove the Lemma by using Eq.~\eqref{Eq:relative_with_max_mixed} in Eq.~\eqref{MRE} and rearranging the terms to obtain 
\begin{equation}
    H(B)_{\Lambda(\rho)} - H(A)_{\rho} \geq \log_2 \frac{\textrm{dim}(B)}{\textrm{dim}(A)}.
    \end{equation}
\end{proof}


The result presented above shows that quantum operations preserving maximally mixed states (with larger output dimensions than input ones) imply the monotonicity of quantum entropy. The interventional strategy set in Theorem~\ref{main_thm} is adopted in order to make the time evolution of the environment to implement a local transformation preserving maximally mixed states.
It is convenient to consider the following result before discussing the proof of Theorem~\ref{main_thm}.

\begin{figure}
    \centering
    \includegraphics[width=\columnwidth]{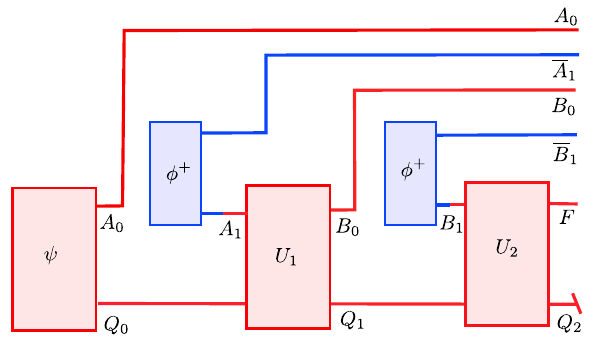}
    \caption{Supermap described in Lemma~\ref{lemma_inf} with the intervention by maximally entangled state $\phi^+_X:= 1/\dim(X)\sum_{i,j}\ketbra{i}{j}_X\otimes \ketbra{i}{j}_{\overline{X}}$ at systems $X\in\{A_1, B_1\}$, where $X\simeq \overline{X}$ and $\{\ket{i}\}_i$ is an orthonormal basis at $X$. Note after the interventions, the overall process becomes isomorphic to a pure state $\tau \in A_0{\otimes}\overline{A}_1{\otimes}B_0{\otimes}\overline{B}_1{\otimes}F{\otimes}Q_2$.}
    \label{fig:Lemma2}
\end{figure}

\begin{lemma} \label{lemma_inf} Let $\Upsilon$ be a non-Markovian process for which holds
\begin{equation}
\Upsilon(\mathrm{A},\mathrm{B})=\Tr_{Q_2}[\mathrm{U}_2 (\mathrm{B} \otimes \mathrm{id})
\mathrm{U}_1 (\mathrm{A} \otimes \mathrm{id}) (\psi)],
\end{equation}
with pure state preparation $\psi$ for $A_0 \otimes Q_0$, and unitary operations $\mathrm{U}_1:L(A_1 \otimes Q_0)\rightarrow L(B_0 \otimes Q_1)$ and $\mathrm{U}_2:L(B_1 \otimes Q_1)\rightarrow L(F \otimes Q_2)$, see Fig.~\ref{fig:Lemma2}. Then, for the intervention by the maximally entangled states at $A_1$ and $B_1$ as defined in Thm.~\ref{main_thm}, the information inequality in Eq.~(\ref{ineq_ent}) holds, i.e.,
\begin{align}
 H(A_0A_1B_0B_1F)_\tau - H(A_0A_1B_0)_\tau \ge \log_2 \frac{\dim(B_1)}{\dim(F)},
\end{align}
where $\tau \in A_0{\otimes}\overline{A}_1{\otimes}B_0{\otimes}\overline{B}_1{\otimes}F{\otimes}Q_2$ is the resulting pure state after the interventions. Note given systems $X$ and $\overline{X}$ are isomporphic: $X\simeq \overline{X}$, we write the entropies $H(X)$ and $H(\overline{X})$ interchangeably.
\end{lemma}

\begin{proof}
Let us denote the state fed at the input system $Q_1$ of $U_2$ by $\eta_{Q_1}$ and denote its purification by the state $\widetilde{\tau}_{A_0\overline{A}_1B_0Q_1}$, i.e., $\eta_{Q_1}:=\Tr_{A_0\overline{A}_1B_0}\widetilde{\tau}$. Then we have the following chain of equalities:
\begin{align}
    H(Q_1)_{\eta}=H(A_0A_1B_0)_{\widetilde{\tau}}=H(A_0A_1B_0)_\tau.\label{Eq:lemma2_1}
\end{align}
Here, the first equality is because  $H(X)_{\alpha}=H(Y)_{\alpha}$ holds for every pure bipartite quantum state $\alpha\in L(X\otimes Y)$. The second equality is due to the fact that the systems $A_0 \otimes \overline{A}_1\otimes B_0$ are going through the identity map. 

Now let us focus on the channel $\Lambda_{Q_1\to Q_2}$ defined as
\begin{align}
 \Lambda(\cdot)&:=\Tr_{F}\circ U_2[ (\cdot)\otimes \Tr_{\overline{B}_1} \phi^+_{B_1\overline{B}_1}] \nonumber \\
& =\Tr_{F}\circ U_2[ (\cdot)\otimes \omega_{B_1}]\label{Eq:compleltely_factorizable_channel_Lemma2}.
\end{align}
One can check, this channel maps a maximally mixed state $\omega_{Q_1}$ to a maximally mixed state $\omega_{Q_2}$. Channels defined as in Eq.\eqref{Eq:compleltely_factorizable_channel_Lemma2} are called \emph{completely factorizable channels}, \cite[Definition 16.1]{gour2024resourcesquantumworld}. Hence, we can apply Lemma~\ref{lemma_dpi} on this channel. Now we have the following chain of equations:
\begin{align}
    &H(A_0A_1B_0B_1F)_\tau = H(Q_2)_{\tau} \nonumber \\
    &= H(Q_2)_{\Lambda(\eta)} \ge  H(Q_1)_{\eta} + \log_2\frac{\dim(Q_2)}{\dim(Q_1)} \nonumber \\
    &= H(A_0A_1B_0)_{\tau} + \log_2\frac{\dim(B_1)}{\dim(F)}.
\end{align}
The first equality is due to $\tau$ being a pure state. The inequality is due to Lemma~\ref{lemma_dpi}. In the last equality we have used Eq.~\eqref{Eq:lemma2_1} and the fact that $U_2$ being unitary, we have $\dim(B_1) \dim (Q_1) = \dim(F) \dim(Q_2)$. Rearranging the above inequality, we obtain the desired Eq.~\eqref{ineq_ent} and prove Lemma~\ref{lemma_inf}.



\end{proof}

It follows a final result which is relevant to our discussion.
In the following, we show that every bipartite quantum non-Markovian process can be cast as a sequence of pure state preparation $\psi \in A_0 \otimes Q_0$, unitary transformations $\mathrm{U}_1: L(A_1 \otimes Q_0 ) \rightarrow L(B_0 \otimes Q_1)$ and $\mathrm{U}_2: L(B_1 \otimes Q_1) \rightarrow L(F \otimes Q_2)$, followed by a partial trace of the final environment system. That is done by taking the purification of states and the dilation of quantum operations. 

\begin{figure}[h!] 
\centering     \includegraphics[width=.45\textwidth]{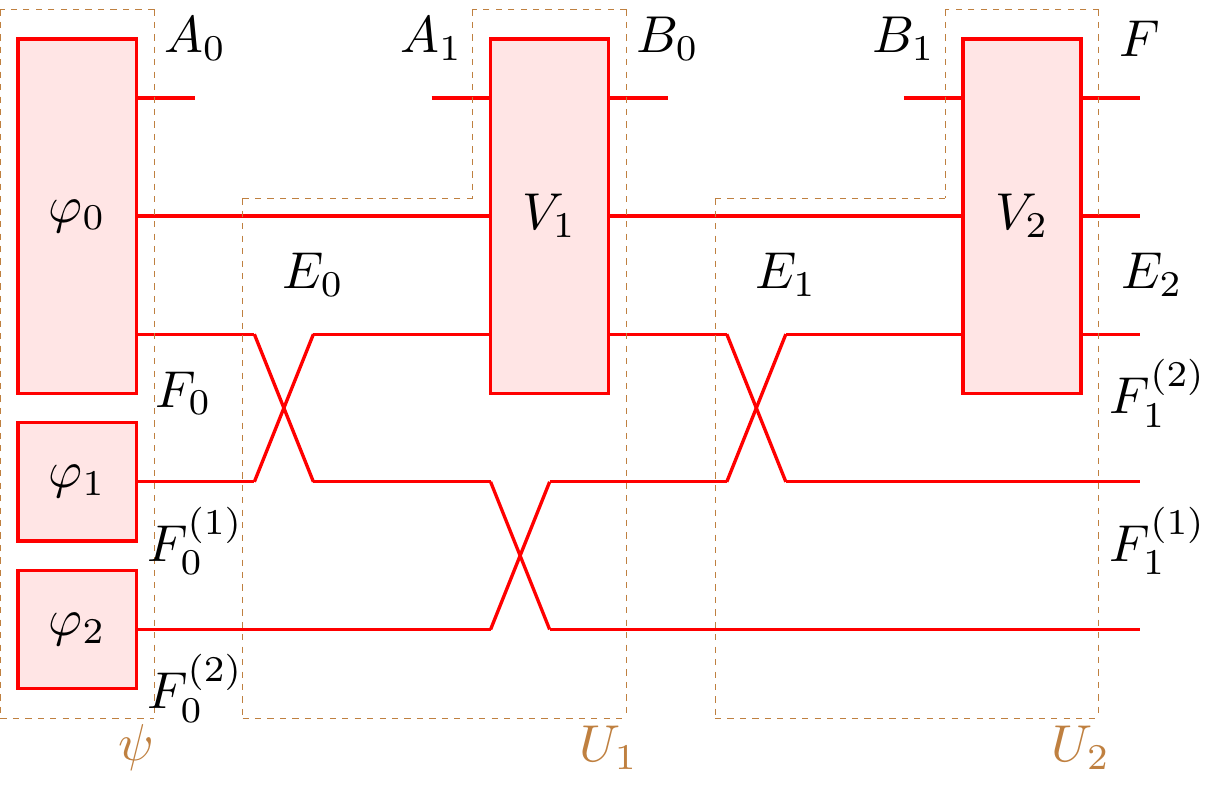} 
    \caption{Purification of a non-Markovian process. Every non-Markovian process is a sequential action of bipartite system-environment quantum operations on an initial state preparation (see top picture $(a)$ in Fig.~\ref{figure2}). Furthermore, we can take a purification of the initial quantum state of the process and Stinespring representations of the quantum channels in order to define an equivalent quantum process on the local laboratories of Alice and Bob. Swapping the ancillary environmental states, which could be absorbed into the unitary operations, results in a sequential use of global unitary operations acting on an initial pure state.}
    \label{fig_appendix_A}
\end{figure}

\begin{lemma} \label{lemma_pure_nm} Let $\Upsilon$ be a non-Markovian quantum process with $A \preceq B$. Then, it holds that
\begin{eqnarray} \label{pure_nm}    \Upsilon(\mathrm{A},\mathrm{B})=\Tr_{Q_2}[\mathrm{U}_2 (\mathrm{B} \otimes \mathrm{id})
\mathrm{U}_1 (\mathrm{A} \otimes \mathrm{id}) (\psi)]
\end{eqnarray}
where $\psi$ is a pure quantum state of a bipartite quantum system $A_0 \otimes Q_0$,  $\mathrm{U}_1:L(A_1 \otimes Q_0)\rightarrow L(B_0 \otimes Q_1)$ and $\mathrm{U}_2:L(B_1 \otimes Q_1)\rightarrow L(F \otimes Q_2)$ are unitary quantum operations.
\end{lemma}

The condition in Eq.~(\ref{pure_nm}) shows that in the definition of non-Markovianity in Eq.~(\ref{non-Markov}) suffices to consider pure state preparations (instead of mixed states) and unitary operations (instead of general quantum operations).

\begin{proof} By definition, the process $\Upsilon$ take the form in Eq.~(\ref{non-Markov}). Furthermore, there is a purification $\varphi_0$ of ${A_0 \otimes E_0 \otimes F_0}$ for $\rho$. Note that $F_0$ works as a purification system. Then, consider a dilation of the quantum operation $\Lambda_1$ with respect to the unitary operation $\mathrm{V}_1:L(A_1 \otimes E_0 \otimes F_{0;1}) \rightarrow L(B_0 \otimes E_1 \otimes F_{1;1})$ and pure state $\varphi_1$ of ${F_{0;1}}$. Similarly, take a dilation for $\Lambda_2$ with $\mathrm{V}_2:L(B_1 \otimes E_1 \otimes F_{0;2}) \rightarrow L(B_0 \otimes E_1 \otimes F_{1;2})$ and pure state $\varphi_1$ of ${F_{0;2}}$. 

Now, consider the definitions
\begin{eqnarray}
    \psi&=&\varphi_0 \otimes \varphi_1 \otimes \varphi_2,\\
    \mathrm{U}_1&=&\mathrm{SWAP}_{F_0:F_{0;2}}(\mathrm{V}_1 \otimes \mathds{1}_{F_{0;1}F_{0;2}})\mathrm{SWAP}_{F_0:F_{0;1}},\\
    \mathrm{U}_2&=&(\mathrm{V}_2 \otimes \mathds{1}_{F_{1;1}F_0})\mathrm{SWAP}_{F_{1;1}:F_{2;0}},
\end{eqnarray}
where the swap operation between subsystems $X$ and $Y$ is denoted with $\mathrm{SWAP}_{X:Y}$.

Figure~\ref{fig_appendix_A} is a diagrammatic proof that Eq.~(\ref{non-Markov}) holds with $Q_0=E_0\otimes F_0 \otimes F_{0;1} \otimes F_{0;2}$, $Q_1=E_1\otimes F_{1;1} \otimes F_{0;2} \otimes F_0$, and $Q_2=E_2\otimes F_{1;2} \otimes F_{1;1} \otimes F_0$.
    
\end{proof}

\noindent \textit{Proof of Theorem \ref{main_thm}:} It follows directly from Lemma~\ref{lemma_inf} along with Lemma~\ref{lemma_pure_nm}.

\end{document}